\documentclass[11pt]{article}
\usepackage[utf8]{inputenc}
\usepackage{fullpage}
\usepackage{setspace}
\usepackage{etex}
\usepackage{latexsym}
\usepackage{amssymb}
\usepackage{mathrsfs}
\usepackage{fancyhdr}
\usepackage{xcolor}
\usepackage{pictex}
\usepackage{placeins}
\usepackage{float}
\usepackage{booktabs}
\usepackage{supertabular}
\usepackage{pict2e}
\usepackage{slashbox}
\usepackage[english]{babel}
\usepackage{graphicx}

\newcommand{\omt}[1]{}

\title{Fast Isomorphism Testing of Graphs with Regularly-Connected Components}
\author{Jos\'e Luis L\'opez-Presa\\
DIATEL, Universidad Polit\'ecnica de Madrid\\
Madrid, Spain\\
jllopez@diatel.upm.es
\and
Antonio Fern\'andez Anta\\
Institute IMDEA Networks\\
Madrid, Spain\\
antonio.fernandez@imdea.org
}
\date{}

\newenvironment{packed_enum}{
\begin{enumerate}
  \setlength{\itemsep}{1pt}
  \setlength{\parskip}{0pt}
  \setlength{\parsep}{0pt}
}{\end{enumerate}}

\newenvironment{packed_itemize}{
\begin{itemize}
  \setlength{\itemsep}{1pt}
  \setlength{\parskip}{0pt}
  \setlength{\parsep}{0pt}
}{\end{itemize}}

\begin{document}

\newtheorem{theorem}{Theorem}
\newtheorem{lemma}{Lemma}
\newtheorem{definition}{Definition}
\newtheorem{proposition}{Proposition}
\newtheorem{corollary}{Corollary}
\newtheorem{conjecture}{Conjecture}
\newtheorem{observation}{Observation}
\newtheorem{remark}{Remark}
\newenvironment{proof}{\noindent{\bf Proof:}}{\hfill \rule{2mm}{2mm}\\}

\newcounter{algline}

\newcommand{\nl}{\\ \>\arabic{algline}\' \stepcounter{algline}}

\newenvironment{algo}[1]{\setcounter{algline}{1}\begin{tabbing} 123\=\=678\=123\=678\=123\=678\=123\=678\=123\=678\=  \kill {\bf #1} \nl}{\end{tabbing}}

\newenvironment{boxfig}[1]{\begin{figure}[!tb]\fbox{\begin{minipage}{0.98\linewidth}
                        \vspace{1em}
                        \makebox[0.025\linewidth]{}
                        \begin{minipage}{0.95\linewidth}
                        #1
                        \end{minipage}
                        \end{minipage}}}{\end{figure}}

\newcommand{\B}{\vspace*{-\smallskipamount}}
\newcommand{\BB}{\vspace*{-\medskipamount}}
\newcommand{\BBB}{\vspace*{-\bigskipamount}}

\floatstyle{ruled}
\newfloat{Algorithm}{!tb}{loa}

\maketitle

\begin{abstract}
The Graph Isomorphism problem has both theoretical and practical interest. In this paper we present
an algorithm, called \emph{conauto-1.2}, that efficiently tests whether two graphs are isomorphic, 
and finds an isomorphism if they are. This algorithm is an improved version of the
algorithm \emph{conauto}, which has been shown to be very fast for random graphs and several families of hard graphs
\cite{DBLP:conf/wea/Lopez-PresaA09}.
In this paper we establish a new theorem that allows, at very low cost, the easy discovery of many
automorphisms. This result is especially suited for graphs with regularly connected components, 
and can be applied in any isomorphism testing and canonical labeling algorithm to drastically 
improve its performance. In particular, algorithm conauto-1.2 is obtained by the application of this result to conauto. 
The resulting algorithm preserves all the nice features
of conauto, but drastically improves the testing of graphs with regularly connected components.
We run extensive experiments, which show that the most popular algorithms (namely, \emph{nauty} \cite{McK81,nautyP} and \emph{bliss} \cite{DBLP:conf/alenex/JunttilaK07})
can not compete with conauto-1.2 for these graph families.
\end{abstract}

\newpage

\section{Introduction}
\label{s1}

The Graph Isomorphism problem (GI) is of both theoretical and practical interest. GI
tests whether there is a one-to-one mapping between the vertices of two graphs that preserves the arcs. 
This problem has applications in many fields, like
pattern recognition and computer vision \cite{ConteFSV03}, data mining \cite{WashioM03},
VLSI layout validation \cite{MAF90}, and chemistry \cite{Faulon98,TinKlin99}. At the theoretical
level,  its main theoretical interest is that it is not known whether GI is in P or whether it is NP-complete.

\paragraph{Related Work}
It would be nice to find a complete graph-invariant\footnote{A complete graph-invariant is
a function on a graph that gives the same result for isomorphic graphs, and different results
for non-isomorphic graphs.} computable in polynomial time,
what would allow testing graphs for isomorphism in polynomial time. However, no
such invariant is known, and it is unlikely to exist. Note, though, that there
are many simple instances of GI, and that many families of graphs
can be tested for isomorphism in polynomial time: trees \cite{AHU74}, planar graphs
\cite{HopWong74}, graphs of bounded degree \cite{FilMay80}, etc. For a review
of the theoretical results related to GI see \cite{DBLP:conf/wea/Lopez-PresaA09,jllopez2009}.

The most interesting practical approaches to the GI problem are (1)
the direct approach, which uses backtracking to find a match between the graphs, using
techniques to prune the search tree, and (2) computing a 
\emph{certificate}\footnote{A certificate of a graph is a canonical labeling of the graph.} of each of the graphs to
test, and then compare the certificates directly.
The direct approach can be used for both graph and subgraph isomorphism (e.g. vf2
\cite{Cor01vf2} and Ullman's \cite{Ullman76} algorithms), but has problems when dealing with
highly regular graphs with a relatively small automorphism group. In this case, even
the use of heuristics to prune the search space frequently does not prevent the proposed algorithms from
exploring paths equivalent to those already tested. To avoid this, it
is necessary to keep track of discovered automorphisms, and use this information to
aggressively prune the search space.
On the other hand, using certificates, since two isomorphic graphs
have the same canonical labeling, their certificates can be compared directly. This is the approach
used by the well-known algorithm \emph{nauty} \cite{McK81,nautyP}, and the algorithm
\emph{bliss} \cite{DBLP:conf/alenex/JunttilaK07} (which has better performance than nauty
for some graph families). This approach requires computing the full automorphism group of the
graph (at least a set of generators). In
most cases, these algorithms are faster than the ones that use the direct approach.

Algorithm \emph{conauto} \cite{DBLP:conf/wea/Lopez-PresaA09} uses a new 
approach to graph isomorphism\footnote{A preliminary version of conauto has 
been included in the LEDA C++ class library of algorithms \cite{singler2005}.}.
It combines the use of discovered
automorphisms with a backtracking algorithm that tries to find a match of the graphs
without the need of generating a canonical form. To test graphs of $n$ nodes
conauto uses $O(n^2 \log n)$ bits of memory. Additionally, it runs in polynomial time
(on $n$) with high probability for random graphs. In real experiments, for several families of interesting
hard graphs, conauto is faster than nauty and vf2, as shown 
in \cite{DBLP:conf/wea/Lopez-PresaA09}. For example
Miyazaki's graphs \cite{Miyazaki97}, are very hard for vf2, nauty, and bliss, but conauto
handles them efficiently.
However, it was found in \cite{DBLP:conf/wea/Lopez-PresaA09} that some families 
of graphs built from regularly connected components (in particular, 
from strongly regular graphs) are not handled efficiently by any of the algorithms evaluated.
While conauto runs fast when the tested graphs are isomorphic, 
it is very slow when the graphs are not isomorphic.

\paragraph{Contributions}
In this paper we establish a new theorem that allows, at very low cost, the easy discovery of many
automorphisms. This result is especially suited for graphs with regularly connected components,
and can be applied in any direct isomorphism testing or canonical 
labeling algorithm to drastically improve its performance.

Then, a new algorithm, called \emph{conauto-1.2}, is proposed. This algorithm is obtained by
improving conauto with techniques derived from the above mentioned theorem. In particular,
conauto-1.2 reduces the backtracking needed to explore every
plausible path in the search space with respect to conauto.
The resulting algorithm preserves all the nice features
of conauto, but drastically improves the testing of some graphs, like those
with regularly connected components.

We have carried out experiments to compare the practical performance of conauto-1.2,
nauty, and bliss, with different families of graphs built by regularly
connecting copies of small components. The experiments show that, for this type
of construction, conauto-1.2 not only is the fastest, but also has a very regular behavior.

\paragraph{Structure}
In Section~\ref{definitions}, we define the basic theoretical concepts used in algorithm conauto-1.2
and present the theorems on which its correction relies. Next, in Section~\ref{conauto-1.2}
we describe the algorithm itself. Then, Section~\ref{performance} describes the graph families
used for the tests, and show the practical performance of conauto-1.2 compared with conauto,
nauty and bliss for these families. Finally we put forward our conclusions and propose
new ways to improve conauto-1.2.

\section{Theoretical Foundation}
\label{definitions}

\subsection{Basic Definitions}

A {\em directed graph} $G=(V,R)$ consists of a finite non-empty set $V$ of vertices and a binary
relation $R$, i.e. a subset $R \subseteq V \times V$.
The elements of $R$ are called {\em arcs}. An arc $(u,v) \in R$ is considered
to be oriented from $u$ to $v$. An {\em undirected graph} is a graph whose arc set $R$ is
symmetrical, i.e. $(u,v) \in R$ iff $(v,u) \in R$. From now on, we will use the term {\em graph}
to refer to a {\em directed graph}.

\begin{definition}
An isomorphism of graphs $G=(V_G,R_G)$ and $H=(V_H,R_H)$ is a bijection between the vertex sets
of $G$ and $H$, $f:V_G \longrightarrow V_H$, such that $(v,u) \in R_G \iff (f(v),f(u)) \in R_H$.
Graphs $G$ and $H$ are called {\em isomorphic}, written $G \simeq H$, if there is at least one
isomorphism of them. An {\em automorphism} of $G$ is an isomorphism of $G$ and itself.
\end{definition}

Given a graph $G=(V,R)$, $R$ can be represented by an {\em adjacency matrix}
$\mathit{Adj}(G)=A$ with size $|V| \times |V|$ in the following way:
\begin{displaymath}
A_{uv} = \left\{ \begin{array}{ll | ll}
        0 & \textrm{if $(u,v) \notin R \land (v,u) \notin R$} & 1 & \textrm{if $(u,v) \notin R \land (v,u) \in R$}\\
        2 & \textrm{if $(u,v) \in R \land (v,u) \notin R$} & 3 & \textrm{if $(u,v) \in R \land (v,u) \in R$}
\end{array} \right.
\end{displaymath}


%
Let $G=(V,R)$ be a graph, and $\mathit{Adj}(G)=A$ its adjacency matrix.
Let $V_1 \subseteq V$ and $v\in V$, the {\em available degree} of $v$ in $V_1$ under $G$,
denoted by $\mathit{ADeg}(v,V_1,G)$, is  the degree of $v$ with respect to $V_1$, i.e., the 3-tuple $(D_3,D_2,D_1)$ where 
$D_i = | \{ u \in V_1: A_{vu} = i \} |$ for $i \in \{1,2,3\}$. The predicate $\mathit{HasLinks}(v,V_1,G)$
says if $v$ has any neighbor in $V_1$, i.e., $\mathit{ADeg}(v,V_1,G) \neq (0,0,0)$.
%
%
Extending the notation, let $V_1,V_2 \subseteq V$; if
$\forall u,v \in V_1, \mathit{ADeg}(u,V_2,G)=\mathit{ADeg}(v,V_2,G)=d$, then, we
denote $\mathit{ADeg}(V_1,V_2,G)=d$. $\mathit{HasLinks}(V_1,V_2,G)$ is defined similarly.


We will say a 3-tuple $(D_3,D_2,D_1) \prec (E_3,E_2,E_1)$ when the first one precedes
the second one in lexicographic order.
This notation will be used
to order the available degrees of vertices and sets.




\subsection{Specific Notation and Definitions for the Algorithms}

It will be necessary to introduce some specific notation to be used in the specification
of our algorithms. Like other isomorphism testing algorithms, ours relies on vertex
classification. Let us start defining what a partition is, and the partition
concatenation operation.

A {\em partition} of a set $S$ is a sequence $\mathcal{S} = ( S_1, ..., S_r )$ of disjoint
nonempty subsets of $S$ such that $S = \bigcup_{i=1}^r S_i$.  The sets $S_i$ are called
the {\em cells} of $\mathcal{S}$. The empty partition will be denoted by $\emptyset$. 

\begin{definition}
Let $\mathcal{S}=(S_1,...,S_r)$ and $\mathcal{T}=(T_1,...,T_s)$ be partitions of two disjoint
sets $S$ and $T$, respectively. The {\em concatenation} of $\mathcal{S}$ and $\mathcal{T}$,
denoted $\mathcal{S} \circ \mathcal{T}$, is the partition $(S_1,...,S_r,T_1,...,T_s)$.
Clearly, $\emptyset \circ \mathcal{S} = \mathcal{S} = \mathcal{S} \circ \emptyset$.
\end{definition}




Let $G=(V,R)$ be a graph, $v \in V$, $V_1 \subseteq V \setminus \{v\}$.
The {\em vertex partition} of $V_1$ by $v$, denoted $\mathit{PartitionByVertex}(V_1,v,G)$, is a
partition $(S_1,...,S_r)$ of $V_1$ such that for all $i,j \in \{1,...,r\}$, $i>j$ implies
$\mathit{ADeg}(S_i,\{v\},G) \prec \mathit{ADeg}(S_j,\{v\},G)$. Let $V_1,V_2 \subseteq V$. The
{\em set partition} of $V_1$ by $V_2$,
denoted $\mathit{PartitionBySet}(V_1,V_2,G)$, is a partition $(S_1,...,S_r)$ of $V_1$ such that
for all $i,j \in \{1,...,r\}$, $i>j$ implies $\mathit{ADeg}(S_i,V_2,G) \prec \mathit{ADeg}(S_j,V_2,G)$.

\begin{definition}
Let $G=(V,R)$ be a graph, and $\mathcal{S}=(S_1,...,S_r)$ a partition of $V$. Let $v \in S_x$
for some $x \in \{1,...,r\}$. The {\em vertex refinement} of $\mathcal{S}$ by $v$, denoted
$\mathit{VertexRefinement}(\mathcal{S},v,G)$ is the partition $\mathcal{T}=\mathcal{T}_1
\circ...\circ \mathcal{T}_r$ such that for all $i \in \{1,...,r\}$, $\mathcal{T}_i$ is the
empty partition $\emptyset$ if $\lnot \mathit{HasLinks}(S_i,V,G)$, and $\mathit{PartitionByVertex}
(S_i\setminus\{v\},v,G)$ otherwise. $S_x$ is the \emph{pivot set} and $v$ is the \emph{pivot vertex}.
\end{definition}


\begin{definition}
Let $G=(V,R)$ be a graph, and $\mathcal{S}=(S_1,...,S_r)$ a partition of $V$. Let $P=S_x$ for
some $x \in \{1,...,r\}$ be a given \emph{pivot set}. The {\em set refinement} of $\mathcal{S}$ by $P$,
denoted $\mathit{SetRefinement}(\mathcal{S},P,G)$ is the partition $\mathcal{T}=\mathcal{T}_1
\circ...\circ \mathcal{T}_r$ such that for all $i \in \{1,...,r\}$, $\mathcal{T}_i$ is the
empty partition $\emptyset$ if $\lnot \mathit{HasLinks}(S_i,V,G)$, and $\mathit{PartitionBySet}
(S_i,P,G)$ otherwise.
\end{definition}

Once we have presented the possible partition refinements that may be applied to partitions, we
can build sequences of partitions in which an initial partition (for example the one with one cell
containing all the vertices of a graph) is iteratively refined using the two previously defined
refinements.  Vertex refinements are tagged as $\mathrm{VERTEX}$
(if the pivot set has only one vertex), $\mathrm{SET}$ (if a set refinement is possible with some
pivot set), or $\mathrm{BACKTRACK}$ (when a vertex refinement is performed with a pivot set with
more than one vertex).

\begin{definition}
Let $G=(V,R)$ be a graph. A {\em sequence of partitions} for graph $G$ is a tuple
$(\mathsf{S},\mathsf{R},\mathsf{P})$, where $\mathsf{S}=(\mathcal{S}^0,...,\mathcal{S}^t)$,
are the partitions themselves, $\mathsf{R}=(R^0,...,R^{t-1})$ indicate the type of
refinement applied at each step, and $\mathsf{P}=(P^0,...,P^{t-1})$ choose the pivot set
used for each refinement step, such that all the following statements hold:
\begin{packed_enum}
\item
For all $i\in\{0,...,t-1\}$, $R^i \in \{\mathrm{VERTEX},\mathrm{SET},\mathrm{BACKTRACK}\}$,
and $P^i \in \{1,...,|\mathcal{S}^i|\}$.
\item
For all $i\in\{1,...,t-1\}$, let $\mathcal{S}^i=(S^i_1,...,S^i_{r_i})$, $V^i=\bigcup_{j=1}^{r_i} S^i_j$. Then:
\begin{packed_enum}
\item
$R^i=\mathrm{SET}$ implies $\mathcal{S}^{i+1} = \mathit{SetRefinement}(\mathcal{S}^i, S^i_{P^i},G)$.
\item
$R^i\ne\mathrm{SET}$ implies $\mathcal{S}^{i+1} = \mathit{VertexRefinement}(\mathcal{S}^i, v,G)$
for some $v \in S^i_{P^i}$.
\end{packed_enum}
\item
Let $\mathcal{S}^t=(S^t_1,...,S^t_r)$, $V^t=\bigcup_{j=1}^r S^t_j$, then for all
$S^t_x\in\mathcal{S}^t$, $|S^t_x|=1$ or $\lnot\mathit{HasLinks}(S^t_x,V^t,G)$.
\end{packed_enum}
\end{definition}

For convenience, for all $l \in \{1,...,t-1\}$, by {\em level} $l$ we refer to the
tuple $(\mathcal{S}^l,R^l,P^l)$ in a sequence of partitions. Level $t$ is identified
by $\mathcal{S}^t$, since $R^t$ and $P^t$ are not defined.

We will now introduce the concept of compatibility among partitions, and then define compatibility
of sequences of partitions.
%
Let $\mathcal{S}=(S_1,...,S_r)$ be a partition of the set of vertices of a graph $G=(V_G,R_G)$, and
let $\mathcal{T}=(T_1,...,T_s)$ be a partition of the set of vertices of a graph $H=(V_H,R_H)$.
$\mathcal{S}$ and $\mathcal{T}$ are said to be {\em compatible} under $G$ and $H$ respectively
if $|\mathcal{S}|=|\mathcal{T}|$ (i.e. $r=s$), and for all $i\in\{1,...,r\}$, $|S_i|=|T_i|$
and $\mathit{ADeg}(S_i,V_G,G) = \mathit{ADeg}(T_i,V_H,H)$.


\begin{definition}
\label{def-compat-seq}
Let $G=(V_G,R_G)$ and $H=(V_H,R_H)$ be two graphs.
Let $\mathsf{Q}_G=(\mathsf{S}_G,\mathsf{R}_G,\mathsf{P}_G)$, 
and $\mathsf{Q}_H=(\mathsf{S}_H,\mathsf{R}_H,\mathsf{P}_H)$
be two sequences of partitions for graphs $G$ and $H$ respectively.
$\mathsf{Q}_G$ and $\mathsf{Q}_H$ are said to be {\em compatible sequences of partitions} if:
\begin{packed_enum}
\item
$|\mathsf{S}_G|=|\mathsf{S}_H|=t$, $|\mathsf{R}_G|=|\mathsf{R}_H|=|\mathsf{P}_G|=|\mathsf{P}_H|=t-1$.
\item
Let $\mathsf{R}_G=(R_G^0,...,R_G^{t-1})$, $\mathsf{R}_H=(R_H^0,...,R_H^{t-1})$,
$\mathsf{P}_G=(P_G^0,...,P_G^{t-1})$, $\mathsf{P}_H=(P_H^0,...,P_H^{t-1})$,
$\mathsf{S}_G=(\mathcal{S}^0,...,\mathcal{S}^t)$, $\mathsf{S}_H=(\mathcal{T}^0,...,\mathcal{T}^t)$. 
For all $i \in \{0,...,t-1\}$, $R_G^i=R_H^i$, $P_G^i=P_H^i$, and $\mathcal{S}^i$ and $\mathcal{T}^i$ are
compatible under $G$ and $H$ respectively.
\item
Let $ \mathcal{S}^t=(S^t_1,...,S^t_r)$, $\mathcal{T}^t=(T^t_1,...,T^t_r)$, then for all
$x,y \in \{1,...,r\}$, $\mathit{ADeg}(S^t_x,S^t_y,G) = \mathit{ADeg}(T^t_x,T^t_y,H)$.
%
\end{packed_enum}
\end{definition}

The following theorem shows that having compatible sequences of partitions is equivalent to being isomorphic.

\begin{theorem}[\cite{DBLP:conf/wea/Lopez-PresaA09}]
\label{iso-iif-seq}
Two graphs $G$ and $H$ are isomorphic if and only if there are two compatible sequences of partitions
$\mathsf{Q}_G$ and $\mathsf{Q}_H$ for graphs $G$ and $H$ respectively.
\end{theorem}

In order to properly handle automorphisms, sequences of partitions will be extended with vertex equivalence
information.
Two vertices $u,v\in V$ of a graph $G=(V,R)$ are {\em equivalent}, denoted $u \equiv v$, if there is
an automorphism $f$ of $G$ such that $f(u)=v$. A vertex $w \in V$ is {\em fixed}
by $f$ if $f(w)=w$.
When two vertices are equivalent, they are said to belong to the same {\em orbit}. The set of all
the orbits of a graph is called the {\em orbit partition}. 
Our algorithm performs a partial computation of the orbit partition. The orbit partition will be
computed incrementally, starting from the singleton partition. Since our algorithm performs a
limited search for automorphisms, it is possible that it stops before the orbit partition is
really found. Therefore, we will introduce the notion of {\em semiorbit partition},
and extend the sequence of partitions to include a semiorbit partition.

\begin{definition}
\label{def-semiorbit-partition}
Let $G=(V,R)$ be a graph. 
A {\em semiorbit partition} of $G$ is any partition $\mathsf{O}=\{O_1,...,O_k\}$
of $V$, such that $\forall i\in \{1,...,k\}$, $v,u \in O_i$ implies that $v \equiv u$.
\end{definition}

%
\begin{definition}
An {\em extended sequence of partitions} $\mathsf{E}$ for a graph
$G=(V,R)$ is a tuple $(\mathsf{Q},\mathsf{O})$, where $\mathsf{Q}$ is a sequence of partitions,
denoted as $\mathit{SeqPart}(\mathsf{E})$, and $\mathsf{O}$ is a semiorbit partition of $G$,
denoted as $\mathit{Orbits}(\mathsf{E})$.
\end{definition}

Finally, we introduce a notation for the number of vertex refinements tagged $\mathrm{BACKTRACK}$,
since it will be used to choose the target
sequence of partitions to be reproduced.
%
Let $\mathsf{Q}=(\mathsf{S},\mathsf{R},\mathsf{P})$ be a sequence of partitions, and let
$\mathsf{R}=(R^0,...,R^{t-1})$. Then,
$\mathit{BacktrackAmount}(\mathsf{Q})=|\{i:i\in \{1,...,t-1\}\land R^i=\mathrm{BACKTRACK}\}|$.

\subsection{Components Theorem}
\label{theorem}

It was observed \cite{DBLP:conf/wea/Lopez-PresaA09} that conauto is very efficient finding
isomorphisms for unions of strongly regular graphs, but it is inefficient detecting that two
such unions are not isomorphic. Exploring the behavior of conauto in graphs that are the
disjoint union of connected components, we observed that it was not able to identify cases
in which components in both graphs had already been matched. This was leading to many redundant
attempts of matching components.

Note that, once a component $C_G$ of a graph $G$ has been found isomorphic to a component $C_H$
of a graph $H$, it is of no use trying to match $C_G$ to another component of $H$. Besides, if
$C_G$ can not be matched to any component of $H$, it is of no use trying to match the other components,
since, at the end, the graphs can not be isomorphic. 
After a thorough study of the behavior of conauto for these graphs, we have
concluded that its performance can be drastically improved in these cases by directly
applying the following theorem (whose proof can be found in the Appendix):

\begin{theorem}
\label{main-theorem}
During the search for a sequence of partitions compatible with the target, backtracking
from a level $l$ to a level $k<l$, such that each cell of level $l$ is contained in a
different cell of level $k$, can not provide a compatible partition.
\end{theorem}

\section{Conauto-1.2}
\label{conauto-1.2}

In this section we propose a new algorithm conauto-1.2 (described in Algorithm~\ref{are-isomorphic}) which
is based on algorithm conauto \cite{DBLP:conf/wea/Lopez-PresaA09}, and uses the result of Theorem~\ref{main-theorem} to drastically reduce
backtracking. It starts generating a sequence of partitions for each of the graphs being tested (using
function $\mathit{GenerateSequenceOfPartitions}$), and performing a limited search for automorphisms using function
$\mathit{FindAutomorphisms}$, just like conauto. The difference with conauto is that, during the search
for the compatible sequence of partitions ($\mathit{Match}$), the algorithm not always backtracks to the
previous recursive call (the previous level in the sequence of partitions). Instead, it may backtrack directly
to a much higher level, or even stop the search, concluding that the graphs are not isomorphic, skipping
intermediate backtracking points.

Function $\mathit{GenerateSequenceOfPartitions}$ is the same used by conauto (see \cite{DBLP:conf/wea/Lopez-PresaA09}
for the details). It is worth mentioning that it generates a sequence of partitions with the following criteria:
\begin{packed_enum}
\item
It starts with the degree partition, and ends when it gets a partition in which no non-singleton cell
has remaining links.
\item
The \emph{pivot cell} used for a refinement must always have remaining links (the more, the better).
\item
At each level, a vertex refinement with a singleton pivot cell is the preferred choice.
\item
The second best choice is to perform a set refinement, preferring small cells over big ones.
\item
If the previous refinements can not be used, then a vertex is chosen from the pivot cell
(the smallest cell with links), a vertex refinement is performed with that \emph{pivot vertex},
and a backtracking point arises.
\end{packed_enum}

Function $\mathit{FindAutomorphisms}$ is also the same used by conauto (see \cite{DBLP:conf/wea/Lopez-PresaA09}
for the details). It takes as input a sequence
of partitions for a graph, and generates an extended sequence of partitions. In the process, it tries
to eliminate backtracking points, and builds a semiorbit partition of the vertices with the information
on vertex equivalences it gathers. Recall that two vertices are equivalent if there is an automorphism that permutes
them, i.e., if there are two equivalent sequences of partitions in which one takes the place of the other.

\begin{Algorithm}{
\caption{Test whether $G$ and $H$ are isomorphic ({\em conauto-1.2}).}
\label{are-isomorphic}
\begin{footnotesize}
\begin{algo}{$\mathit{AreIsomorphic}(G,H):\mathrm{boolean}$}
\> $\mathsf{Q}_G \leftarrow \mathit{GenerateSequenceOfPartitions}(G)$ \nl 
$\mathsf{Q}_H \leftarrow \mathit{GenerateSequenceOfPartitions}(H)$ \nl

\> $\mathsf{E}_G \leftarrow \mathit{FindAutomorphisms}(G,\mathsf{Q}_G)$ \nl
$\mathsf{E}_H \leftarrow \mathit{FindAutomorphisms}(H,\mathsf{Q}_H)$ \nl
\> {\bf if} $\mathit{BacktrackAmount}(\mathit{SeqPart}(\mathsf{E}_G)) \leq \mathit{BacktrackAmount}(\mathit{SeqPart}(\mathsf{E}_H))$ {\bf then}\nl
\>\> {\bf return} $0\le\mathit{Match}(0,G,H,\mathit{SeqPart}(\mathsf{E}_G),\mathit{Orbits}(\mathsf{E}_H))$ \nl
\> {\bf else} \nl
\>\> {\bf return} $0\le\mathit{Match}(0,H,G,\mathit{SeqPart}(\mathsf{E}_H),\mathit{Orbits}(\mathsf{E}_G))$ \nl
\> {\bf end if}
\end{algo}
\end{footnotesize}
\BB\BB
}
\end{Algorithm}

Function $\mathit{Match}$ (Algorithm~\ref{match}) uses backtracking attempting to find a sequence of partitions
for graph $H$ that is compatible with the one for graph $G$. At backtracking points, it tries every feasible
vertex in the pivot cell, so that no possible solution is missed.

Note that, unlike in conauto, the function $\mathit{Match}$ of conauto-1.2 does not return a boolean, but an
integer. Thus, if
$\mathit{Match}$ returns $-1$, that means that a mismatch has been found at some level $l$, such that there
is no previous level $l'$ at which a cell contains (at least) two cells of the partition of level $l$. Hence,
from Theorem~\ref{main-theorem} there is no other feasible alternative in the search space that can yield an
isomorphism of the graphs. If it returns a value that is higher than the current level, then a match has been
found, the graphs are isomorphic and there is no need to continue the search. Therefore, in this case the
call immediately returns with this value. If it returns a value that is lower than the current level, then it is
necessary to backtrack to that level, since trying another option at this level is meaningless according to
Theorem~\ref{main-theorem}. Hence the algorithm also returns immediately with that value. If a call at level
$l$ returns $l$, then another alternative at this level $l$ should be tried if possible. In any other case,
it applies Theorem~\ref{main-theorem} directly, and returns the closest (previous) level $l'$ at which two cells
of the current level $l$ belong to the same cell of $l'$. If no such previous level exists, it returns $-1$.

\begin{Algorithm}{
\caption{Find a sequence of partitions compatible with the target.}
\label{match}
\begin{footnotesize}
\begin{algo}{$\mathit{Match}(l,G,H,\mathsf{Q}_G,\mathsf{O}_H):\mathrm{integer}$}
\> {\bf if} partition labeled $\mathrm{FIN}$ {\bf \em and} the  adjacencies in both partitions match\nl
\>\> return $l$ \nl
\> {\bf else if} partition labeled $\mathrm{VERTEX}$ {\bf \em and} vertex refinement compatible {\bf then} \nl
\>\> $l' \longleftarrow \mathit{Match}(l+1,G,H,\mathsf{Q}_G,\mathsf{O}_H)$ \nl
\>\> {\bf if} $l \ne l'$ {\bf then} {\bf return} $l'$ \nl
\> {\bf else if} partition labeled $\mathrm{SET}$ {\bf \em and} set refinement compatible {\bf then} \nl
\>\> $l' \longleftarrow \mathit{Match}(l+1,G,H,\mathsf{Q}_G,\mathsf{O}_H)$ \nl
\>\> {\bf if} $l \ne l'$ {\bf then} {\bf return} $l'$ \nl
\> {\bf else if} partition labeled $\mathrm{BACKTRACK}$ {\bf then} \nl
\>\> {\bf for each} vertex $v$ in the pivot cell, while NOT success {\bf do} \nl
\>\>\> {\bf if} $v$ may NOT be discarded according to $\mathsf{O}_H$ {\bf \em and} vertex refinement compatible {\bf then} \nl
\>\>\>\> $l' \longleftarrow \mathit{Match}(l+1,G,H,\mathsf{Q}_G,\mathsf{O}_H)$ \nl
\>\>\>\> {\bf if} $l \ne l'$ {\bf then} {\bf return} $l'$ \nl
\>\>\> {\bf end if} \nl
\>\> {\bf end for} \nl
\> {\bf end if} \nl
\> {\bf return} the nearest level $l'$ such that the condition of Theorem~\ref{main-theorem} holds
\end{algo}
\end{footnotesize}
\BB\BB
}
\end{Algorithm}

\section{Performance Evaluation}
\label{performance}

In this section we compare the practical performance of conauto-1.2 with nauty and bliss, two
well-known algorithms that are considered the fastest algorithms for isomorphism testing and canonical
labeling. In the performance evaluation experiments, we have run these programs with instances
(pairs of graphs) that belong to specific families. We also use conauto to show the improvement
achieved by conauto-1.2 for these graph families. Undirected and directed (when possible) graphs of different sizes
(number of nodes) have been considered. The experiments include instances of isomorphic and
non-isomorphic pairs of graphs.

\subsection{Graph Families}
\label{graph-cons}

For the evaluation, we have built some families of graphs with regularly-connected 
components. The general construction technique of these graphs consists of
combining small components of different types by either (1) connecting every vertex of
each component to all the vertices of the other components, (2) connecting only some
vertices in each component to some vertices in all the other
components, or (3) applying the latter construction in two levels.
The use of these techniques guarantees that the resulting graph is connected, which
is convenient to evaluate algorithms that require connectivity (like, e.g., vf2 \cite{Cor01vf2}).
Using the disjoint union of connected components yields similar experimental results.

Next, we describe each
family of graphs used. In fact, as the reader will easily infer, the key point in
all these constructions is that the components are either disconnected, or connected
via complete $n$-partite graphs. Hence, multiple other constructions may be used
which would yield similar results. In each graph family, one hundred pairs of isomorphic and
non-isomorphic graphs have been generated for each graph size (up to approximately $1,000$ vertices).

\paragraph{Unions of Strongly Regular Graphs}
\label{USR-graphs}

This graph family is built from a set of $20$ strongly regular graphs with parameters
$(29,14,6,7)$ as components. The components are interconnected so that each vertex in one component
is connected to every vertex in the other components. This is equivalent to inverting
the components, then applying the disjoint union, and finally inverting the result. Graphs
up to $20\times29=580$ vertices have only one copy of each component, and bigger ones
may have more than one copy of each component. 

\paragraph{Unions of Tripartite Graphs}
\label{tripartite-graphs}

For this family, we use the digraphs in Figure~\ref{tripartitos-fig} as the basic
components. For the positive tests (isomorphic graphs) we use the same number of components
of each type, while for the negative tests we use one graph with the same number of components
of each type, and another graph in which one component has been replaced by one of the
other type.

The connections between components have been done in the following way. The vertices in the
$A$ subset of each component are connected to all the vertices in the $B$ subsets of the
other components. See Figure~\ref{tripartitos-fig} to locate these subsets. The arcs are
directed from the vertices in the $A$ subsets, to the vertices in the $B$ subsets.
From the previously described graphs, we have obtained an undirected version by transforming
every (directed) arc into an (undirected) edge. 

\begin{figure*}[tb!]
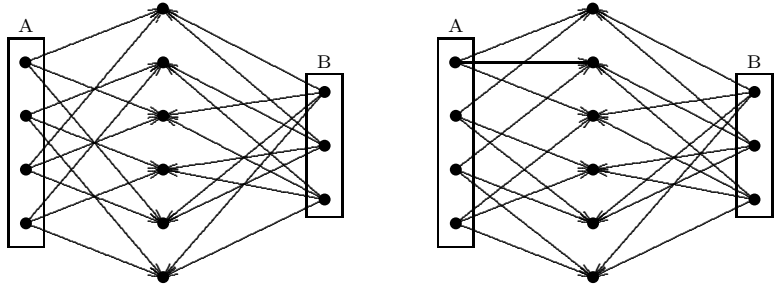

\centerline{\input tripartitos.pictex}
\caption{Tripartite graphs used as components.}
\label{tripartitos-fig}
\end{figure*}

\paragraph{Hypo-Hamiltonian Graphs 2-level-connected}
\label{CHH-graphs}

For this family we use two non-isomorphic Hypo-Hamiltonian graphs with 22 vertices. Both graphs
have four orbits
of sizes: one, three, six, and twelve. These basic
components are interconnected at two levels. Let us call the vertices in the orbits of size one,
the {\em 1-orbit} vertices, and the vertices in the orbits of size three the {\em 3-orbit} vertices.
In the first level, we connect $n$ basic components, to form a {\em first-level component},
by connecting all the $3$-orbit vertices in each basic component to all the $3$-orbit vertices
of the other basic components. In this construction, the $3$-orbit vertices, along with the
new edges added to interconnect the $n$ basic components, form a complete $n$-partite graph.
Then, in the second level, $m$ first-level components are interconnected by adding edges that connect the $1$-orbit
vertices of each first-level component with all the $1$-orbit vertices of the other first-level
components. Again, the $1$-orbit vertices, along with the edges connecting them, form a complete
$m$-partite graph.
Since we use two Hypo-Hamiltonian graphs as basic components, to generate negative isomorphism cases,
a component of one type is replaced with one of the other type.


\subsection{Evaluation Results}

The performance of the four programs has been evaluated in terms of their execution time with multiple instances
of graphs from the previously defined families.  The execution times have been measured in a Pentium III at 1.0
GHz with 256 MB of main memory, under Linux RedHat $9.0$. The same compiler (GNU gcc) and the same optimization
flag (-O) have been used to compile all the programs. The time measured is the real execution time (not only CPU
time) of the programs. This time does not include the time to load the graphs from disk into memory. A time limit
of $10,000$ seconds has been set for each execution. When the execution of a program with graphs of size $s$ reaches
this limit, all the execution data of that program for graphs of the same family with size no smaller than $s$ are
discarded.

\paragraph{Average Execution Time}

The results of the experiments are first presented, in Figure~\ref{average}, as curves that represent execution
time as a function of graph size. In these curves, each point is the average execution time of the corresponding
program on all the instances of the corresponding size.

\begin{figure*}[tb!]
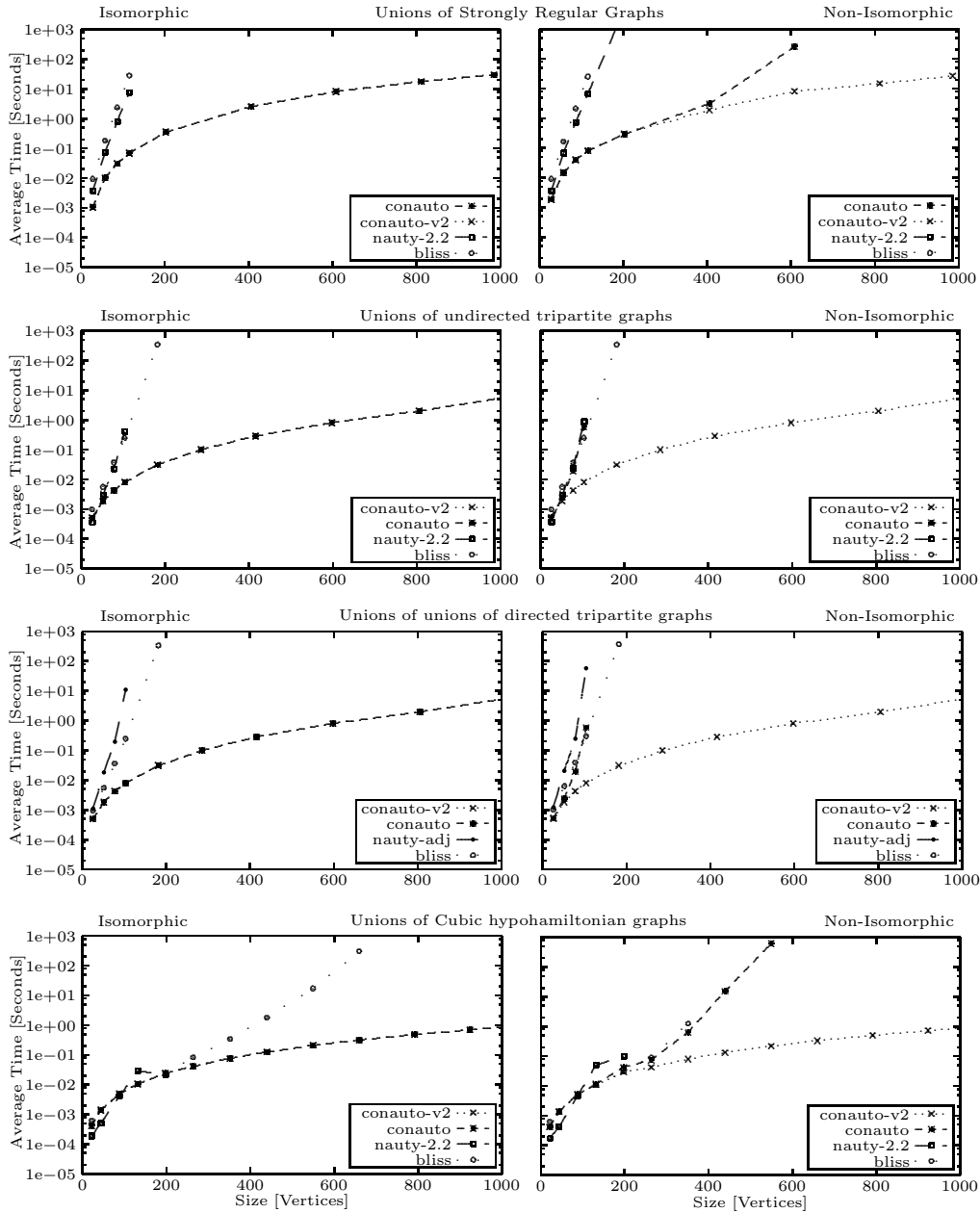

\centerline{\input media.pictex}
\caption{Average execution time.}
\label{average}
\end{figure*}


It was previously known that nauty requires exponential time to process graphs that are unions of strongly
regular graphs \cite{Miyazaki97}. From our results, we conjecture that bliss has the same problem. That does
not apply to conauto-1.2, though. While the original conauto had problems with non-isomorphic pairs of graphs,
conauto-1.2 overcomes this problem.



With the family of unions of tripartite graphs, we have run both positive and negative experiments with directed
and undirected versions of the graphs. In all cases, conauto-1.2 has a very low execution time. 
(Again, the improvement of conauto-1.2 over conauto is apparent in the case of negative tests.)
Observe that there are no significant differences in the execution times of bliss and conauto-1.2 between the directed
and the undirected cases. However, nauty is slower with directed graphs, even using the \emph{adjacencies} invariant
specifically designed for directed graphs.

Our last graph family, Cubic Hypohamiltonian 2-level-connected graphs, has a more complex structure than the other families, 
having two levels of interconnection. However, the results do
not differ significantly from the previous ones. It seems that these graphs are a bit easier to process (compared
with the other graph families) for bliss, but not for nauty. Like in the previous cases, conauto-1.2
is fast and consistent with the graphs in this family. It clearly
improves the results of conauto for the non-isomorphic pairs of graphs.


\paragraph{Standard Deviation}

In addition to the average behavior for each graph size, we have also evaluated the regular behavior of the
programs. With regular behavior we mean that the time required to process
any pair of graphs of the same family and size is very similar. 
We have observed that conauto-1.2 is not only fast for all these families of graphs,
but it also has a very regular behavior.
However, that does not hold for nauty
nor bliss. This is illustrated with the plots of the \emph{normalized standard deviation}\footnote{The
normalized standard deviation is obtained by dividing the standard deviation of the sample by the mean.}
(NSD) shown in Figure~\ref{desv}. Algorithm conauto-1.2 has a NSD that remains almost constant, and very
close to cero, for all graph sizes, and even decreases for larger graphs. However, nauty and bliss have
a much more erratic behavior. In the case of conauto, we see that its problems arise when it faces negative
tests, where the NSD rapidly grows.

\begin{figure*}[tb!]
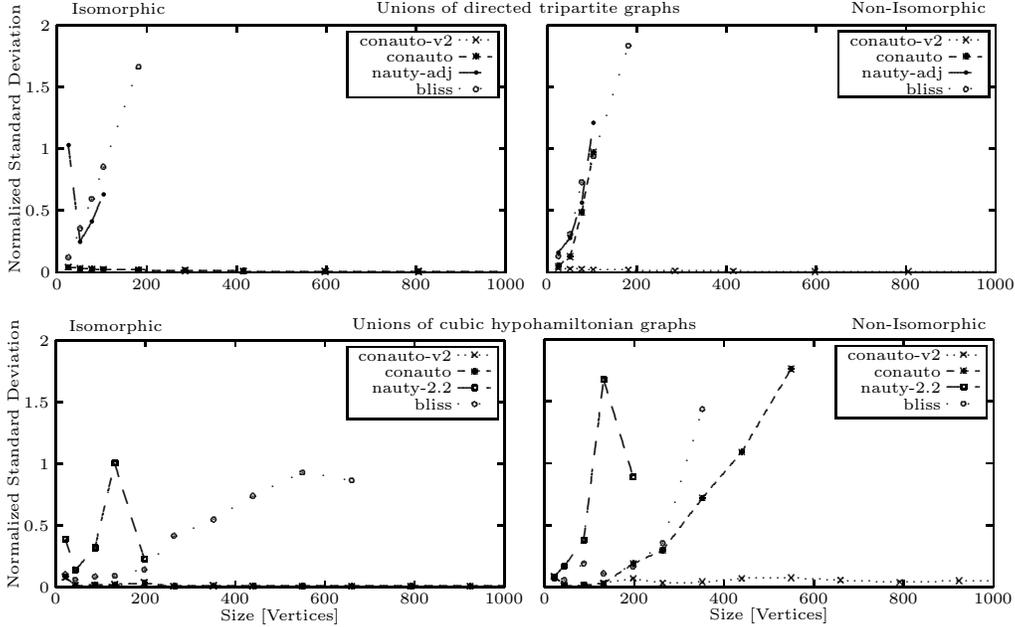

\centerline{\input desv.pictex}
\caption{Normalized Standard Deviation of execution times.}
\label{desv}
\end{figure*}



\section{Conclusions and Future Work}
\label{future}

We have presented a result (the Components Theorem, Theorem~\ref{main-theorem}) that can be applied in GI algorithms to efficiently find automorphisms. 
Then, we have applied this result to transform the algorithm conauto into conauto-1.2.
Algorithm conauto-1.2 has been shown to be fast and consistent in performance for a variety of graph families.
However, the algorithm conauto-1.2 can still be improved
in several ways: (1) by adding the capability of computing a complete set of generators for the automorphism
group, (2) by making extensive use of discovered automorphisms during the match process, and (3) by computing
canonical forms of graphs. In all these possible improvements, the Components Theorem will surely
help. Additionally, the Components Theorem might also be used by nauty and bliss to improve their
performance for the graph families considered, at low cost.


\newpage
\appendix

\section{Proof of the Components Theorem (Theorem~\ref{main-theorem})}

The following definition will be needed in the proof.

\begin{definition}
Let $G=(V,R)$ be a graph. Let $V'\subseteq V$. Then the {\em subgraph induced}
by $V'$ on $G$, denoted $G_{V'}$, is the graph $H=(V',R')$ such that
$R'=\{(u,v) : u,v \in V' \land (u,v) \in R\}$.
\end{definition}

A backtracking point arises when a partition does not have singleton cells (suitable
for a vertex refinement) and it is not possible to refine such partition by means of
a set refinement. Let us introduce a new concept that will be useful in the following
discussion.

\begin{definition}
\label{def-equitable}
Let $G=(V,R)$ be a graph, and let $\mathcal{S}=(S_1,...,S_r)$ be a partition of
$V$. $\mathcal{S}$ is said to be {\em equitable} (with respect to $G$) if for all
$i \in \{1,...,r\}$, for all $u,v \in S_i$, for all $j \in \{1,...,r\}$,
$\mathit{ADeg}(u,S_j,G) = \mathit{ADeg}(v,S_j,G)$.
\end{definition}

\begin{observation}
\label{sk-equitable}
The partition at a backtracking point is equitable.
\end{observation}

\begin{proof}
Assume otherwise. Then, there exists some $S_j$ such that there are two vertices $u,v$ in some
$S_i$, such that $\mathit{ADeg}(u,S_j,G) \ne \mathit{ADeg}(v,S_j,G)$. Therefore,
it would be possible to perform a set refinement on the partition, using $S_j$ as
the pivot cell, and vertices $u$ and $v$ would be distinguished by this refinement,
and cell $S_i$ would be split. This is not possible since, at a backtracking point,
no set refinement has succeeded.
\end{proof}

\begin{observation}
\label{G-S-l-x-regular}
Let $l$ be a backtracking level. Let $\mathcal{S}^l=(S^l_1,...,S^l_r)$ be the partition at
that level. Then, for all $i \in \{1,...,r\}$, $G_{S^l_i}$ is regular.
\end{observation}

\begin{proof}
From Observation \ref{sk-equitable}, $\mathcal{S}^l$ is equitable. Fix $i \in \{1,...,r\}$, then,
from Definition \ref{def-equitable}, for all $u,v \in S^l_i$, $\mathit{ADeg}(u,S^l_i,G)=\mathit{ADeg}(v,S^l_i,G)$.
Therefore, $G_{S^l_i}$ is regular, for all $i \in \{1,...,r\}$.
\end{proof}

Let $\mathsf{Q}=(\mathsf{S},\mathsf{R},\mathsf{P})$ be a sequence of partitions
for graph $G=(V,R)$ where $\mathsf{S}=(\mathcal{S}^0,...,\mathcal{S}^t)$,
$\mathsf{R}=(R^0,...,R^{t-1})$, and $\mathsf{P}=(P^0,...,P^{t-1})$. For all $i \in \{0,...,t\}$
let $\mathcal{S}^i=(S^i_1,...,S^i_{r_i})$, and $V^i=\bigcup_{j=1}^{r_i} S^i_j$.
We consider two backtracking levels $k$ and $l$ that satisfy the preconditions of
Theorem~\ref{main-theorem}, i.e., $k<l$ and each cell of $\mathcal{S}^l$ is contained
in a different cell of $\mathcal{S}^k$.

Let $p \in S^k_{P^k}$ be the pivot vertex used for the vertex refinement at level $k$.
Assume there is a vertex $q \in S^k_{P^k}, q \ne p$ that satisfies the following.
$\mathcal{T}^{k+1}=\mathit{VertexRefinement}(\mathcal{S}^k,q,G_{V^k})$ is a partition
that is compatible with $\mathcal{S}^{k+1}$.  Let $\mathcal{T}^{k+1}=(T^{k+1}_1,...,T^{k+1}_{r_{k+1}})$,
$W^{k+1}=\bigcup_{j=1}^{r_{k+1}} T^{k+1}_j$.  For all $i \in \{k+2,...,l\}$, let
$\mathcal{T}^i=(T^i_1,...,T^i_{r_i})$ be compatible with $\mathcal{S}^i$, where
$W^i=\bigcup_{j=1}^{r_i} T^i_j$,
$\mathcal{T}^i=\mathit{SetRefinement}(\mathcal{T}^{i-1},T^{i-1}_{P^{i-1}},G_{W^{i-1}})$
if $R^{i-1}=\mathrm{SET}$, and $\mathcal{T}^i=\mathit{VertexRefinement}(\mathcal{T}^{i-1},v,G_{W^{i-1}})$
for some $v \in T^{i-1}_{P^{i-1}}$ if $R^{i-1} \ne \mathrm{SET}$.
This generates an alternative sequence of partitions that is compatible with the original
one up to level $l$.

Under these premises, we show in the rest of the section that $G_{V^l}$ and $G_{W^l}$
are isomorphic, and there is an isomorphism of them that matches the vertices in $S^l_i$
to the vertices in $T^l_i$ for all $i\in\{1,...,r_l\}$.

To simplify the notation, let us assume $r_k=r_l=r$. Note that in this case, for all
$i\in\{1,...,r\}$, $S^l_i \subseteq S^k_i$. In case $r_k \ne r_l$ this correspondence
is not trivial. However, we can safely assume that there may be some $S^l_i \in \mathcal{S}^l$
that are empty, and develop our argument considering this possibility, although we
know that in the real sequence of partitions, these empty cells would have been discarded.

For all $i\in\{1,...,r\}$, let $E_i=S^k_i \setminus S^l_i$, $E'_i=S^k_i \setminus T^l_i$
be the vertices discarded in the refinements from $S^k_i$ to $S^l_i$ and $T^l_i$ respectively,
let $A_i=E_i\cap E'_i$ be the vertices discarded in both alternative refinements,
$B_i=E_i \setminus A_i$ the vertices discarded only in the refinement from $S^k_i$ to $S^l_i$,
$C_i=E'_i \setminus A_i$ the vertices discarded only in the refinement from $S^k_i$ to $T^l_i$,
and $D=S^l_i \cap T^l_i$ the vertices remaining in both alternative partitions at level $l$.
Let $A=\bigcup_{i=1}^r A_i$, $B=\bigcup_{i=1}^r B_i$, $C=\bigcup_{i=1}^r C_i$, $D=\bigcup_{i=1}^r D_i$,
$E=\bigcup_{i=1}^r E_i$, and $E'=\bigcup_{i=1}^r E'_i$. Clearly, $E=A\cup B$, and $E'=A\cup C$.
Observe that $|E_i|=|E'_i|$, and hence $|B_i|=|C_i|$ for all $i\in\{1,...,r\}$.

\begin{figure}[!h]
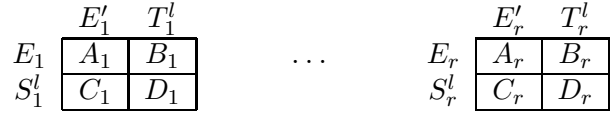

\begin{center}
\begin{tabular}{c|c|c|}
 \multicolumn{1}{c}{ } & \multicolumn{1}{c}{$E'_1$} & \multicolumn{1}{c}{$T^l_1$} \\
\cline{2-3} \multicolumn{1}{c|}{$E_1$} & $A_1$ & $B_1$ \\
\cline{2-3} \multicolumn{1}{c|}{$S^l_1$} & $C_1$ & $D_1$ \\
\cline{2-3}
\end{tabular}
\hspace{1cm}
{\large\ldots}
\hspace{0,7cm}
\begin{tabular}{c|c|c|}
 \multicolumn{1}{c}{ } & \multicolumn{1}{c}{$E'_r$} & \multicolumn{1}{c}{$T^l_r$} \\
\cline{2-3} \multicolumn{1}{c|}{$E_r$} & $A_r$ & $B_r$ \\
\cline{2-3} \multicolumn{1}{c|}{$S^l_r$} & $C_r$ & $D_r$ \\
\cline{2-3}
\end{tabular}
\end{center}
\label{division-S-k-i}
\caption{Partition of $S^k_i$ into subsets $A_i$, $B_i$, $C_i$, and $D_i$ for all $i\in\{1,...,r\}$.}
\end{figure}

\begin{observation}
\label{Ge-iso-Ge-prima}
$G_E$ is isomorphic to $G_{E'}$, and there is an isomorphism of them that matches the vertices
in $E_i$ to those in $E'_i$, for all $i\in\{1,...,r\}$.
\end{observation}

\begin{proof}
Direct from the construction of the sequences of partitions.
\end{proof}

\begin{lemma}
\label{uno-a-todos}
Let $M=\mathit{Adj}(G)$. It is satisfied that:
\begin{packed_itemize}
\item For each $u \in E$, for all $i\in\{1,...,r\}$, for all $v,w \in S^l_i$, $M_{uv} = M_{uw}$ and $M_{vu} = M_{wu}$.
\item For each $u \in E'$, for all $i\in\{1,...,r\}$, for all $v,w \in T^l_i$, $M_{uv} = M_{uw}$ and $M_{vu} = M_{wu}$.
\end{packed_itemize}
\end{lemma}

\begin{proof}
Since none of the vertices in $E$ has been able to distinguish among the vertices in cell $S^l_i$,
each of the discarded vertices has the same type of adjacency with all the vertices in $S^l_i$.
Otherwise, consider vertex $u \in E$. Assume $u$ has at least two different types of adjacency
with the vertices in $S^l_i$. Since it was discarded during the refinements from $S^k_i$ to $S^l_i$,
that had to be for one of the following reasons:
\begin{packed_enum}
\item
It was discarded for having no links (i.e. links of type $0$), what is impossible since it has two different types of
adjacencies with the vertices in $S^l_i$.
\item
It was used as the pivot set in a vertex refinement, what is impossible since it would have
been able to split cell $S^l_i$.
\end{packed_enum}

The same argument applies to the vertices in $E'$ with respect to the vertices in each cell $T^l_i$.
\end{proof}

Consider the adjacency between vertex $u$ and vertex $v$ is $M_{uv}=a$ for some $a\in\{0,...,3\}$.
Then, we will denote the adjacency between $v$ and $u$ ($M_{vu}$) as $a^{-1}$. Note that if $a=0$,
$a^{-1}=0$, if $a=1$, $a^{-1}=2$, if $a=2$, $a^{-1}=1$, and if $a=3$, $a^{-1}=3$.

\begin{lemma}
\label{Bi-Cj-todos-a-todos}
For each $i,j\in\{1,...,r\}$, there is some $a\in\{0,...,3\}$ such that
for all $u\in B_i$, $v\in C_i$, $w\in D_i$, $u'\in B_j$, $v'\in C_j$, and $w'\in D_j$,
$M_{uv'}=M_{uw'}=M_{vu'}=M_{vw'}=M_{wu'}=M_{wv'}=a$ and $M_{u'v}=M_{u'w}=M_{v'u}=M_{v'w}=M_{w'u}=M_{w'v}=a^{-1}$.
\end{lemma}

\begin{proof}
Let us take any $i\in\{1,...,r\}$ and any $j\in\{1,...,r\}$. Since $B_i\subseteq E$ and
$C_j\subseteq S^l_j$, from Lemma~\ref{uno-a-todos}, for each $u\in B_i$, for all $v'\in C_j$,
$M_{uv'}=a$ for some $a\in\{0,...,3\}$. Let us take any such $v'\in C_j$. Then, $M_{v'u}=a^{-1}$
for those particular $v'$ and $u$. Besides, since $C_j\subseteq E'$ and $B_i\subseteq T^l_i$, from
Lemma~\ref{uno-a-todos}, for all $u\in B_i$, $M_{v'u}=b$ for some $b\in\{0,...,3\}$. Since
we already know that $M_{v'u}=a^{-1}$ for that particular pair of vertices, then we conclude
that for all $u\in B_i$, $v'\in C_j$, $M_{uv'}=a$ and $M_{v'u}=a^{-1}$, for some $a\in\{0,...,3\}$.

$S^l_j=C_j\cup D_j$ and $B_i\subseteq E$. Since for all $u\in B_i$, $v'\in C_j$, $M_{uv'}=a$
and $M_{v'u}=a^{-1}$, then from Lemma~\ref{uno-a-todos}, for all $u\in B_i$, $w'\in D_j$,
$M_{uw'}=a$ (clearly, the same $a$) and $M_{w'u}=a^{-1}$.

$T^l_i=B_i\cup D_i$ and $C_j\subseteq E'$. Since for all $u\in B_i$, $v'\in C_j$, $M_{uv'}=a$
and $M_{v'u}=a^{-1}$, then from Lemma~\ref{uno-a-todos}, for all $v'\in C_j$, $w\in D_i$,
$M_{v'w}=a^{-1}$ and $M_{wv'}=a$ (clearly, the same $a$).

Furthermore, all the vertices in $S^l_j=C_j\cup D_j$ have the same number of adjacent vertices of each
type in $E_i=A_i\cup B_i$. Otherwise, they would have been distinguished in the refinement process
from $\mathcal{S}^k$ to $\mathcal{S}^l$. Likewise, all the vertices in $T^l_j=B_j\cup D_j$ have the
same number of adjacent vertices of each type in $E'_i=A_i\cup C_i$. Otherwise, they would have been
distinguished in the refinement process from $\mathcal{S}^k$ to $\mathcal{T}^l$. Hence, the vertices
of $D_j$ must have the same number of adjacent vertices of each type in $B_i$ and $C_i$. Hence, since
for all $w'\in D_j$, and for all $u\in B_i$, $M_{uw'}=a$ and $M_{w'u}=a^{-1}$, then for all $w'\in D_j$,
and for all $v\in C_i$, $M_{vw'}=a$ and $M_{w'v}=a^{-1}$ too.

A similar argument may be used to prove that for all $w\in D_i$, and for all $u'\in B_j$, $M_{wu'}=a$
and $M_{u'w}=a^{-1}$. Then, from Lemma~\ref{uno-a-todos}, since $B_j\subseteq E$, for all $u'\in B_j$,
$M_{u'x}=M_{u'y}$ for all $x,y\in S^l_i$. We already know that for all $u'\in B_j$, $M_{u'w}=a^{-1}$
for all $w\in D_i$, and $S^l_i=C_i\cup D_i$. Hence, for all $v\in C_i$, $M_{u'v}=a^{-1}$ too, and
$M_{vu'}=a$.

Putting together all the partial results obtained, we get the assertion stated in the lemma.
\end{proof}

\begin{corollary}
\label{B-C-D-todos-a-todos}
Let $M=\mathit{Adj}(G)$. For each $i\in\{1,...,r\}$, it is satisfied that for all
$u\in B_i,v\in C_i,w\in D_i$, $M_{uv}=M_{vu}=M_{uw}=M_{wu}=M_{vw}=M_{wv}=a$, where $a \in \{0,3\}$.
\end{corollary}

\begin{proof}
From Lemma~\ref{Bi-Cj-todos-a-todos}, for the case $i=j$, we get that for all $u\in B_i$,
$v\in C_i$, $w\in D_i$, $M_{uv}=M_{uw}=M_{vu}=M_{vw}=M_{wu}=M_{wv}=a$ and
$M_{uv}=M_{uw}=M_{vu}=M_{vw}=M_{wu}=M_{wv}=a^{-1}$. Hence, it must hold that $a=a^{-1}$, so $a \in \{0,3\}$.
\end{proof}

Let us define two families of partitions of $A_i$ for $i,j\in\{1,...,r\}$:
$$A_i^{cj}=\{x\in A_i : \forall u \in B_i, v' \in C_j, M_{xv'} = M_{uv'} \}$$
$$A_i^{nj}=\{x\in A_i : \forall u \in B_i, v' \in C_j, M_{xv'} \ne M_{uv'} \}$$
Note that, since the vertices of $A_i$ are unable to distinguish among the vertices
of $C_j$, then, if $M_{xv'} \ne M_{uv'}$ for some $u \in B_i$ or some $v' \in C_j$,
then $M_{xv'} \ne M_{uv'}$ for all $u \in B_i$ and all $v' \in C_j$. Hence, each pair
of sets $A_i^{cj}$ and $A_i^{nj}$ defines a partition of $A_i$. Note also that, since
each vertex in $A_i$ has the same type of adjacency with all the vertices in
$B_i \cup C_i \cup D_i$ (from Lemma~\ref{uno-a-todos}), then for all $x \in A_i^{cj}$,
$u\in B_i$, $v\in C_i$, $w\in D_i$, $u'\in B_j$, $v'\in C_j$, and $w'\in D_j$,
$M_{xu'}=M_{xv'}=M_{xw'}=M_{uv'}=M_{uw'}=M_{vu'}=M_{vw'}=M_{wu'}=M_{wv'}$ (from
Lemma~\ref{Bi-Cj-todos-a-todos}).

\begin{figure}[!h]
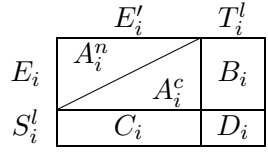

\begin{center}
\begin{tabular}{c|c|c|}
\multicolumn{1}{c}{ } & \multicolumn{1}{c}{$E'_i$} & \multicolumn{1}{c}{$T^l_i$} \\
\cline{2-3} \multicolumn{1}{c|}{$E_i$} & \slashbox{$A^n_i$}{$A_i^c$} & $B_i$ \\
\cline{2-3} \multicolumn{1}{c|}{$S^l_i$} & $C_i$ & $D_i$ \\
\cline{2-3}
\end{tabular}
\end{center}
\label{division-A-i}
\caption{Partition of $A_i$ into subsets $A^c_i$, and $A^n_i$.}
\end{figure}

\begin{lemma}
\label{An-themselves}
For all $i\in\{1,...,r\}$, let $A_i^c=\bigcap_{j=1}^r A_i^{cj}$, and let
$A_i^n=\bigcup_{j=1}^r A_i^{nj}$. Then, any isomorphism of $G_E$ and $G_{E'}$
that maps $G_{E_i}$ to $G_{E'_i}$, maps the vertices in $A_i^n$ among themselves.
\end{lemma}

\begin{proof}
From Observation~\ref{sk-equitable}, partition $\mathcal{S}^k$ is equitable. Hence, for each
$i,j\in\{1,...,r\}$, for all $u,v\in S^k_i$, $\mathit{ADeg}(u,S^k_j,G)=\mathit{ADeg}(v,S^k_j,G)$.
Thus, for all $x \in A_i^{cj}$, $y\in A_i^{nj}$, $u\in B_i$, $v\in C_i$, $w\in D_i$,
$\mathit{ADeg}(x,S^k_j,G)=\mathit{ADeg}(y,S^k_j,G)=\mathit{ADeg}(u,S^k_j,G)=\mathit{ADeg}(v,S^k_j,G)=\mathit{ADeg}(w,S^k_j,G)$.

Let us take any pair of values of $i$ and $j$. From Lemma~\ref{Bi-Cj-todos-a-todos}, all the
vertices of $B_i$ have the same type of adjacency with all the vertices of $S^l_j=C_j\cup D_j$.
Assume this type of adjacency is $a$. From the definition of $A_i^{cj}$, all the vertices of
$A_i^{cj}$ have adjacency $a$ with all the vertices of $S^l_j$.  Hence, for $x \in A_i^{cj}$,
$u\in B_i$, $\mathit{ADeg}(x,S^l_j,G)=\mathit{ADeg}(u,S^l_j,G)$.  Since
$\mathit{ADeg}(x,S^k_j,G)=\mathit{ADeg}(u,S^k_j,G)$ and $\mathit{ADeg}(x,S^l_j,G)=\mathit{ADeg}(u,S^l_j,G)$,
then $\mathit{ADeg}(x,E_j,G)=\mathit{ADeg}(u,E_j,G)$ (note that $E_j=A_j^{ci}\cup A_j^{ni}\cup B_j$,
$S^l_j=C_j\cup D_j$, and $S^k_j=E_j\cup S^l_j$).

However, from the definition of $A_i^{nj}$, for $y\in A_i^{nj}$,
$\mathit{ADeg}(y,S^l_j,G)\ne \mathit{ADeg}(x,S^l_j,G)$.  Hence, since
$\mathit{ADeg}(y,S^k_j,G)=\mathit{ADeg}(x,S^k_j,G)$, $\mathit{ADeg}(y,E_j,G)\ne \mathit{ADeg}(x,E_j,G)$.

Since any isomorphism must match vertices with the same degree,
every isomorphism of $G_E$ and $G_{E'}$ that maps $G_{E_i}$ to $G_{E'_i}$,
maps the vertices in $A^{nj}_i$ among themselves.

Applying this argument over all possible values of $j$, we get that  
any isomorphism of $G_E$ and $G_{E'}$ that maps $G_{E_i}$ to $G_{E'_i}$,
maps the vertices in $A_i^n$ among themselves, for all $i\in\{1,...,r\}$.
\end{proof}

Let us focus on any isomorphism of $G_E$ and $G_{E'}$ that maps $G_{E_i}$ to $G_{E'_i}$
for all $i\in\{1,...,r\}$ (there is at least one from Observation~\ref{Ge-iso-Ge-prima}).

\begin{lemma}
\label{Gb-iso-Gc}
$G_B$ is isomorphic to $G_C$, and there is an isomorphism of them that matches the vertices
in $B_i$ to those in $C_i$, for all $i\in\{1,...,r\}$.
\end{lemma}
\begin{proof}

Let us analyze the adjacencies between the vertices in $A_i^c$, $B_i$, $C_i$, $A_j^c$, $B_j$,
and $C_j$ for some values of $i$ and $j$. From Corollary~\ref{B-C-D-todos-a-todos}, for all
$u\in B_i$, $v\in C_i$, $M_{uv}=M_{vu}=a$, where $a \in \{0,3\}$. From the definition of $A_i^c$,
for all $x\in A_i^c$, $M_{xu}=M_{xv}=M_{ux}=M_{vx}=M_{uv}=a$.

From Lemma~\ref{An-themselves}, the vertices of $A_i^n$ are mapped among themselves in  
any isomorphism of $G_E$ and $G_{E'}$ that maps $G_{E_i}$ to $G_{E'_i}$. Hence,
the vertices of $A_i^c\cup B_i$ must be mapped to the vertices of $A_i^c\cup C_i$. If
$a=0$, then $A^c_i$, $B_i$, and $C_i$ are disconnected. Hence, $G_{B_i}$ and $G_{C_i}$
must be isomorphic. In the case $a=3$, taking the inverses of the graphs leads to the same
result.

From Lemma~\ref{Bi-Cj-todos-a-todos}, for each $i,j\in\{1,...,r\}$,  there is some
$a\in\{0,...,3\}$ such that for all $u\in B_i$, $v\in C_i$, $u'\in B_j$, $v'\in C_j$,
$M_{uv'}=M_{vu'}=a$ and $M_{u'v}=M_{v'u}=a^{-1}$. From the definition of $A_i^c$, for
all $x\in A_i^c$, for all $u\in B_i$, $v\in C_i$, $u'\in B_j$, $v'\in C_j$,
$M_{xu'}=M_{xv'}=M_{uv'}$.

\begin{figure*}[!htb]
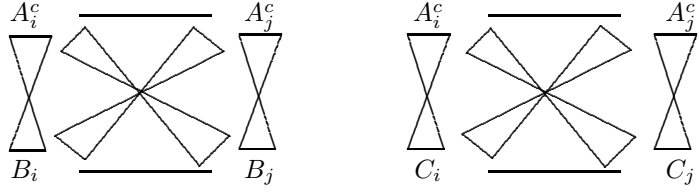

\centerline{\input B_isoC.pictex}
\caption{Adjacencies between $E_i$ and $E_j$, and between $E'_i$ and $E'_j$.}
\label{B_iso_C}
\end{figure*}

Putting all this together, we come to a picture of the adjacencies among $A_i^c$, $B_i$, $C_i$, $A_j^c$,
$B_j$, and $C_j$ as shown in Figure~\ref{B_iso_C}. The connections between the vertices of $A_i^c$ and
the vertices of $B_i$, and between the vertices of $A_i^c$ and the vertices of $C_i$ are all-to-all
(all the same) of value $0$ or $3$. Similarly, the adjacencies between the vertices of $A_j^c$ and the
vertices of $B_j$, and the adjacencies between the vertices of $A_j^c$ and the vertices of $C_j$ are
all the same, all-to-all $0$ or $3$ (not necessarily equal to those of $A_i^c$ and $B_i$ or $C_i$).
The adjacencies between $A_i^c$ and $B_j\cup C_j$ are all the same, all-to-all of any value in the set
$\{0,...,3\}$. This also applies to the adjacencies between $A_j^c$ and $B_i\cup C_i$.

If $G_{B_i\cup B_j}$ is not isomorphic to $G_{C_i\cup C_j}$, the discrepancy must be in the adjacencies
between vertices of $B_i$ and $B_j$ with respect to the adjacencies between vertices of $C_i$ and $C_j$.
In such a case, in the isomorphism between $G_{E_i\cup E_j}$ and $G_{E'_i\cup E'_j}$ (recall that
from Observation~\ref{Ge-iso-Ge-prima} there is an isomorphism of $G_E$ and $G_{E'}$ that maps the vertices
of $E_i$ to the vertices in $E'_i$ for all $i\in\{1,...,r\}$) some vertices of $A_i^c$ should be mapped
to vertices of $C_i$, and some of the vertices of $B_i$ should be mapped to vertices of $A_i^c$. However,
due to the adjacencies among $A_i^c$, $B_i$, $C_i$, $A_j^c$, $B_j$, and $C_j$, shown in Figure~\ref{B_iso_C},
that would imply that the adjacencies between the vertices of $B_i$ and $B_j$ had to match adjacencies
between the vertices of $A^i_c$ and $A^j_c$. But, in that case, the same adjacency pattern must exist
between the vertices of $C_i$ and $C_j$, to match the corresponding subgraph of $G_{E_i\cup E_j}$.
Hence, the adjacencies between $B_i$ and $B_j$ could have been matched to the adjacencies between
$C_i$ and $C_j$.

Since this applies for all values of $i$ and $j$, we conclude that $G_B$ is isomorphic to $G_C$,
and there is an isomorphism of them that matches the vertices in $B_i$ to those in $C_i$, for all
$i\in\{1,...,r\}$, completing the proof.
\end{proof}

\begin{lemma}
\label{GVl-iso-GWl}
$G_{V^l}$ and $G_{W^l}$ are isomorphic, and there is an isomorphism of them that maps the vertices
in $S^l_i$ to the vertices of $T^l_i$ for all $i \in \{1,...,r\}$.
\end{lemma}

\begin{proof}
From Lemma~\ref{Bi-Cj-todos-a-todos}, we know that for each $i,j\in\{1,...,r\}$, there is some
$a\in\{0,...,3\}$ such that for all $u\in B_i$, $v\in C_i$, $w\in D_i$, $u'\in B_j$, $v'\in C_j$,
and $w'\in D_j$, $M_{uv'}=M_{uw'}=M_{vu'}=M_{vw'}=M_{wu'}=M_{wv'}=a$ and
$M_{u'v}=M_{u'w}=M_{v'u}=M_{v'w}=M_{w'u}=M_{w'v}=a^{-1}$.

Note also that, from Corollary~\ref{B-C-D-todos-a-todos},
for all $u\in B_i$, $v\in C_i$, $w\in D_i$, $M_{uv}=M_{vw}=M_{wu}=a$, where $a \in \{0,3\}$.
This adjacency pattern is graphically shown in Figure~\ref{grafica-Gv_iso_Gw}.

\begin{figure*}[htb!]
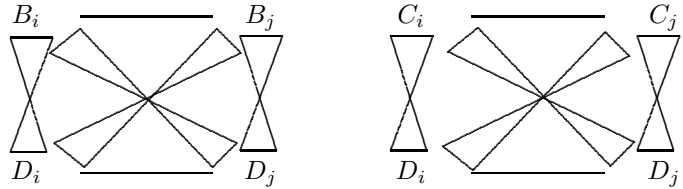

\centerline{\input Gv_iso_Gw.pictex}
\caption{Adjacencies between $S^l_i$ and $S^l_j$, and between $T^l_i$ and $T^l_j$.}
\label{grafica-Gv_iso_Gw}
\end{figure*}

From Lemma~\ref{Gb-iso-Gc}, we know that $G_B$ is isomorphic to $G_C$, and there is an isomorphism
of them that matches the vertices in $B_i$ to those in $C_i$, for all $i\in\{1,...,r\}$.

From the fact that $G_D$ is isomorphic to itself, and the previous considerations on the adjacency
pattern between the vertices in $B_i$, $C_i$, $D_i$, $B_j$, $C_j$, and $D_j$ for all $i,j\in\{1,...,r\}$,
shown in Figure~\ref{grafica-Gv_iso_Gw}, it is easy to see that the isomorphism of $G_B$ and $G_C$
obtained from Lemma~\ref{Gb-iso-Gc}, toghether with the trivial automorphism of $G_D$ yields an
isomorphism of $G_{V^l}$ and $G_{W^l}$, what completes the proof.
\end{proof}

We have shown that if two alternative sequences of partitions $S^{k+1},...,S^l$ and $T^{k+1},...,T^l$
lead to compatible partitions $S^l$ and $T^l$, where all their cells are subcells of different cells of  
a previous common level $k$, then the remaining subgraphs are isomorphic, and the vertices in each cell
of one partition may be mapped to the vertices in its corresponding cell in the other partition by one
such isomorphism. Thus, if during the search for a sequence of partitions compatible with the target,  
we have got an incompatibility at some point beyond level $l$, and we have to backtrack from one level
$l$ to another level $k$ in which all the cells are different supersets of the cells in the current
backtracking point, when trying a compatible path, we will get to the same dead-end. Hence, it is of no
use to try another path from one such level $k$, and it will be necessary to backtrack to some point
where at least two cells in the current backtracking point are subsets of the same cell in the previous  
backtracking point. This proves Theorem~\ref{main-theorem}.

\end{document}